\documentclass[a4paper,11pt,reqno]{amsart}
\usepackage{amsmath,amssymb,amsthm}
\usepackage{dsfont}

\newtheorem{lemma}{Lemma}[section]
\newtheorem{theorem}[lemma]{Theorem}

\newtheorem{proposition}[lemma]{Proposition}
\newtheorem{corollary}[lemma]{Corollary}

\theoremstyle{definition}
\newtheorem{definition}[lemma]{Definition}

\newtheorem*{remark}{Remark}

\DeclareMathOperator{\supp}{supp}

\newcommand{\vL}{\varLambda}
\newcommand{\vG}{\varGamma}

\newcommand{\NN}{\mathbb{N}}

\newcommand{\ZZ}{\mathbb{Z}}

\newcommand{\RR}{\mathbb{R}}
\newcommand{\CC}{\mathbb{C}}

\newcommand{\dd}{\,\mathrm{d}}

\newcommand{\ts}{\hspace{0.5pt}}

\begin{document}
\title{Note on the set of Bragg peaks with high intensity}
%\title{A Note on Point Sets with Many Bragg Peaks  }

\author{Daniel Lenz}
\address{Institut f\"ur Mathematik, Friedrich Schiller Universit\"at Jena, 07743 Jena,
Germany}
\email{daniel.lenz@uni-jena.de}
\urladdr{http://www.analysis-lenz.uni-jena.de}

\author{Nicolae Strungaru}
\address{Department of Mathematical Sciences, MacEwan University
\\
10700 " 104 Avenue, Edmonton, AB, T5J 4S2;\\
and \\
Institute of Mathematics ``Simon Stoilow'' \\
Bucharest, Romania}
\email{strungarun@macewan.ca}
\urladdr{http://academic.macewan.ca/strungarun/}

\begin{abstract}
We consider  diffraction of Delone sets in Euclidean space. We show
that the set of Bragg peaks with high intensity is always Meyer (if
it is relatively dense). We use this to provide  a new
characterization for Meyer sets in terms of positive and positive
definite measures. Our results are based on a careful study of
positive definite measures, which may be of interest in its own
right.
\end{abstract}

\maketitle
\section*{Introduction}
Since the first diffraction of a crystal experiment was performed by
Max von Laue in 1912, physical diffraction has been the most
powerful tool for obtaining  insights in the atomic structure of
crystals.

The diffraction pattern of a fully periodic crystal consists of
bright spots, called {Bragg peaks}, which appear at very precise
locations: on the dual lattice to the lattice of periods of the
crystal.

The diffraction of more general solids is usually a mixture of Bragg
peaks and diffuse background ({continuous spectrum}), with random
structures showing the trivial Bragg peak only.

For a long time it was believed that only periodic crystals can
produce pure point spectra, viz.~a diffraction pattern consisting
exclusively of Bragg peaks. But in 1984, Shechtman et.~all
\cite{She} reported the discovery of a solid with pure point
diffraction and 5-fold symmetry, which is impossible in periodic
crystals. Because of this discovery, the International Union of
Crystallography redefined in 1991 the term of crystal to mean "any
solid having an \textit{essentially discrete} diffraction diagram"
\cite{crystal}.

For an overview of the precise mathematical setup of physical
diffraction we refer the reader to \cite[Chapter~9]{TAO} as well as
to the articles \cite{BG,HOF,Lenz2,Lenz3} (for background on the
physical theory see also \cite{JMC}). The diffraction pattern is
described via the diffraction measure. This measure arises  as the
Fourier transform of the autocorrelation measure  of the point set
(or, more general, measure) which represents the model of the solid.
By Lebesgue decomposition, the  diffraction measure  can be
decomposed into a discrete  part, corresponding to the Bragg
spectrum, and a continuous part, corresponding to the continuous
diffraction spectrum.

It is usually understood that the diffraction is essentially
discrete if the set of Bragg peaks is relatively dense. In this
context special subsets of Euclidean space have been prime examples.
These sets were introduced by Meyer in the 70s and later became
known as Meyer sets.  Indeed, the investigations of Meyer sets has
been central to the topic of diffraction, see e.g. \cite{MOO,BM,BLM}
and references therein. Meyer sets in Euclidean space do have  a
relatively dense set of Bragg peaks as had been suspected for a long
time and was finally shown in \cite{NS1}. Recent work of Kellendonk
/ Sadun \cite{KS} shows even a converse of some sort. More
precisely,  it gives that a dynamical system of Delone sets with
finite local complexity is conjugate to a dynamical system of Meyer
sets if it has a relatively dense set of continuous eigenvalues. In
this sense, Meyer sets seem 'unavoidable' when one deals with  sets
with many Bragg peaks.

The main result of this note  is Theorem \ref{T1} in  Section
\ref{section-main}. It  provides another, somewhat surprising,
instance of this unavoidability of Meyer sets. Namely, we show  for
ANY Delone set in Euclidean space that the  set of Bragg peaks with
high intensity is a Meyer set (provided it is relatively dense).
Therefore, Meyer sets appear in a natural way also in the Bragg
diffraction of any point sets with large sets of Bragg peaks of high
intensity.  If the underlying point set is Meyer itself the
requirement of high intensity can be dropped as has already been
shown in \cite{NS2,NS5}. This can be understood as saying that the
class of Meyer sets is characterized by some form of selfduality
under Fourier transform.  In this spirit, we use our main  result to
give a new characterization for Meyer sets at the end of this note.

Our main result follows  from a more general result dealing with
positive and positive definite measures in  Section
\ref{section-study}. In fact, that section  contains a study of
positive and positive definite measures which may be of interest in
its own right. The necessary background and notation is discussed in
Section \ref{section-background}.

\bigskip

{\large \sc \bf Acknowledgment:} {\small Part of the results
presented in this manuscript where obtained  when Nicolae Strungaru
visited Daniel Lenz and  Peter Stollmann at Chemnitz University some
years ago, and Nicolae would like to take this opportunity to thank
the University for the kind hospitality and Daniel would like to
thank the German Research Foundation (DFG) for partial support.
Nicolae would also like to thank NSERC for their support. Both
authors would like to thank the organizers of the Miniworkshop
'Dynamical and diffraction spectra in the theory of quasicrystals'
at  MFO in 2014, where this work was finished.}

\section{Background and notation}\label{section-background}
In this section we recall some basic concepts underlying our
considerations. We are mainly interested in special subsets of
Euclidean space (or, more generally, of a locally compact abelian
group). However, our statements and proofs can be very conveniently
 phrased in the framework of measures. For this reason we introduce
here both some background on measures and on sets.

\bigskip

For this entire section $G$ denotes a locally compact abelian group
(LCAG). We will denote by $C_{\mathsf{U}}(G)$ the space of uniformly
continuous bounded function on $G$, which is a Banach space with
respect to $\| \cdot \|_\infty$ norm. By  $C_{\mathsf{c}}(G)$ we
denote the subspace of $C_{\mathsf{U}}(G)$ of all compactly
supported functions. We define the convolution $f\ast g$ of
 $f,g\in C_{\mathsf{c}}(G)$ by
 $$f \ast g : G\longrightarrow \CC, x\mapsto \int_G f(x - y) g(y)
 dy,$$
 where $dy$ denotes integration with respect to the Haar measure on
 $G$. Then, $f\ast g$  can easily be seen to belong to $C_{\mathsf{c}}(G)$ as
 well. Moreover, for any complex-valued function $f$ on $G$ we define the function
 $\widetilde{f}$ on $G$ via $\widetilde{f}(x) =\overline{f (-x)}$.

%The definition of the inductive limit topology and
%a direct calculation yield
%\begin{equation*}
%   \sup\, \{ \lvert \mu(\psi)\rvert :
%   \lvert\psi\rvert\le\varphi \}  \; < \; \infty
%\end{equation*}
%for every $\mu\in\MM$ and all $\varphi\in C_c (G)_{+} :=
%\{ \varphi\in C_c (G) :
%    \varphi(t) \ge 0 \mbox{ for all } t\in G \}$.
%Therefore, \cite[Thm.~6.5.6]{Ped} together with its proof
%shows that the mapping
%\begin{equation*}
%   C_c (G)_{+} \; \longrightarrow \; \RR \; , \quad
%   \varphi \, \mapsto \, \sup\, \{ \lvert\mu (\psi) \rvert :
%   \psi \in C_c (G,\RR) \mbox{ with } \lvert\psi\rvert\le\varphi \}
%\end{equation*}
%can be extended to a linear functional on $C_c (G,\RR):=
%\{\varphi\in C_c (G) : \varphi(t)\in\RR \mbox{ for all } t\in G\}$,
%which is uniquely determined.
%This functional can further be extended uniquely to a  linear functional on
%$C_c (G) = C_c(G,\CC)$. This functional
%is denoted by $\lvert\mu\rvert$
%and  called the {\em total variation}\/ of $\mu$.

The space
 $C_{\mathsf{c}}(G)$  is made into a
locally convex space by the inductive limit topology, as induced by
the canonical embeddings
\begin{equation*}
  C^{}_K (G) \; \hookrightarrow \; C_{\mathsf{c}} (G)
  \; , \quad K\subset G \mbox{ compact}.
\end{equation*}
Here, $C_K (G)$ is the space of complex valued continuous functions
on $G$ with support in $K$, which is equipped with the usual
supremum norm.  In line with the Riesz-Markov theorem, for us a
measure $\mu$ on $G$ will then be   a linear functional on
$C_{\mathsf{c}}(G)$, which is continuous with respect to the
inductive topology on $C_{\mathsf{c}}(G)$, see \cite{BL,Ped} for
details.

The convolution of a measure $\mu$ with an $f\in C_{\mathsf{c}}(G)$
is defined as
$$ \mu \ast f : G\longrightarrow \CC, \;\: \mu \ast f (x) = \mu (f(x
- \cdot)).$$

 As is well known
(see e.g. \cite[Thm.~6.5.6]{Ped} together with its proof), every
measure $\mu$  gives rise to a unique measure  $\lvert\mu\rvert$
called the {\em total variation}\/ of $\mu$, satisfying
\begin{equation*}
 \lvert\mu\rvert(f) \; = \; \sup\, \{ \lvert\mu (g) \rvert :
   g \in C_{\mathsf{c}} (G,\RR) \mbox{ with } \lvert g\rvert\leq f\}
\end{equation*}
for every non-negative $f \in C_{\mathsf{c}} (G)$. The total
variation is a positive measure i.e. it maps non-negative functions
to non-negative values and allows for the usual integration theory.
Moreover, by \cite[Thm.~6.5.6]{Ped}, there exists  a measurable
function $u\!: G\longrightarrow \CC$ with $\lvert u(t)\rvert = 1$
for $\lvert\mu\rvert$-almost every $t\in G$ such that
\begin{equation*}
  \mu( f) \; = \; \int_G  f \, u \dd \lvert\mu\rvert
  \quad \mbox{ for all } f \in C_{\mathsf{c}} (G).
\end{equation*}
We can use this to define for any measure $\mu$ on $G$ and any
bounded measurable function $h$ on $G$ the measure $h \mu$ by
$$(h \mu) (f) := \int_G f\, h\, u \, \dd \lvert\mu\rvert$$
 for $f \in C_{\mathsf{c}} (G)$.
We can then also  define the  \textit{discrete part} of a measure
$\mu$ by
$$\mu_{\mathsf{pp}}   = \sum_{p\in P} u(p) \delta_p,$$
where, for $q\in G$, we define the measure  $\delta_q$ via $\delta_q
(f) = f (q)$ and  $P$ is the set of those $q\in G$ such that
$\lvert\mu\rvert (g ) \geq 1$ whenever $ g \in C_{\mathsf{c}} (G)$
is non-negative with $ g (q) =1$.

The measure $\mu$ on $G$  is called  \textit{translation bounded} if
for each compact set\/ $K$ we have
  \[
  \sup_{x \in G} \left| \mu \right| (x+K) \, < \, \infty.
  \]

As mentioned already, besides measures and functions certain subsets
of $G$ with additional properties  are the main object   in our
considerations. The corresponding pieces of notation are introduced
next.

A subset  of\/ $G$ is called\/ {\em uniformly discrete} if there
exists an open set $V$ in $G$ containing the neutral element of $G$
such that $(x + V) \cap (y + V) = \emptyset$ whenever $x$ and $y$
are two different elements of the subset. A subset of  $G$ is
called\/ {\em relatively dense} if there exists a compact $K$ such
that any translate of $K$ intersects the subset. A subset of $G$
which is both uniformly discrete and relatively dense is called {\em
Delone}.  A subset of $G$ is called \textit{weakly uniformly
discrete} if for any compact $K\subset G$ there is a $C$ such that
any translate of $K$ meets at most $C$ points of $\varGamma$.  One
can identify a weakly uniformly discrete set $\vL$ in $G$ with a
measure by considering its \textit{Dirac comb}
$$\delta_\vL:= \sum_{x\in \vL} \delta_x.$$
In this way, all considerations below dealing with measures
naturally descend to weakly uniformly discrete sets and in
particular to Delone sets.

\section{A study of positive and positive definite
measures}\label{section-study}
 We will be interested in measures and
functions with additional positivity properties. More specifically,
we will be interested in positive definite measures. We will present
a study of certain features. As a consequence we will derive three
main properties at end of this section.

\bigskip

\begin{definition} Let $G$ be an LCAG.
\begin{itemize}
   \item The measure $\mu$ on $G$  is \textit{positive definite} if for all $f \in  C_{\mathsf{c}}(G)$ we have
   $\mu  (f*\widetilde{f}) \, \geq \, 0$
\item The function $f : G\longrightarrow \CC$ is \textit{positive
definite} if for all $N\in \NN$ and $x_1,\ldots, x_N\in G$, the
matrix $(f (x_k - x_l))_{k,l=1,\ldots, N}$ is positive definite.
\end{itemize}
\end{definition}
It is well known (see e.g. \cite{BG}) that a measure $\mu$ is
positive definite if and only if  $\mu \ast (f\ast \widetilde{f})$
is positive definite for all $f\in C_{\mathsf{c}}(G)$. Also, we have
the following well known result (\cite{BG}).

\begin{proposition}[Krein's
inequality for functions]\label{Krein-inequality-for-functions} Let
$G$ be an LCAG.  Let $f$ be a positive definite function on $G$.
Then $f(0) \geq |f(x)|$ for all $x\in G$ and
$$|f (x+t) - f(x)|^2 \leq 2 f(0)[f(0) - \Re (f(t)) ]$$
for all $x,t\in G$ (where  $\Re$ denotes the real part).
\end{proposition}

We start by showing that the restriction to the pure point part
preserves positivity and positive definiteness.

\begin{lemma}\label{L-help}
Let $G$ be an LCAG. Let\/ $\mu$ be measure on $G$ and
$\mu_{\mathsf{pp}}$ its discrete part.
\begin{itemize}
  \item[(a)] If\/ $\mu$ is positive then\/ $\mu_{\mathsf{pp}}$ is positive.
  \item[(b)] If\/ $\mu$ is positive definite then\/ $\mu_{\mathsf{pp}}$ is positive definite.
  \item[(c)] If\/ $\mu$ is positive and positive definite then\/ $\mu_{\mathsf{pp}}$ is positive and positive definite.
\end{itemize}
\end{lemma}
\noindent{\bf Proof}: $(a)$ is obvious. $(b)$ follows from
\cite[Thm.~10.2]{ARMA}.  $(c)$ is an immediate consequence of
 $(a)$ and $(b)$. \qed

\bigskip

We are now heading towards a Krein inequality for measures. To
establish it we will need some preparation. We let $G_{\mathsf{d}}$
be the group   $G$ equipped with the discrete topology.  Consider
now a discrete measure $\mu$ on $G$. Then, $\mu$ can be identified
with a measure on $G_{\mathsf{d}}$.

Also, $\mu$ defines a function  on $G$ via $f (x) := \mu(\{x\})$. We
call it the \textit{support function} of $\mu$. In the next
Proposition~\ref{P1}  we show that the positive definiteness of
$\mu$ as measure on $G$ respectively $G_{\mathsf{d}}$ and of $f$ as
function on $G$ respectively $G_{\mathsf{d}}$ are equivalent. This
will allow us to translate properties of positive definite functions
to discrete measures.

\begin{proposition}\label{P1}
Let $G$ be an LCAG. Let\/ $\mu$ be a discrete measure on $G$ and let
$f : G \to \CC$ be its support  function
\[
f(x) \, := \, \mu( \{ x \}) \ts .
\]
Then the following assertions are equivalent:
\begin{itemize}
  \item[$(i)$] $\mu$ is a positive definite measure on $G$.
  \item[$(ii)$] $\mu$ is a positive definite measure on $G_{\mathsf{d}}$.
  \item[$(iii)$] $f$ is a positive definite function on $G$.
  \item[$(iv)$] $f$ is a positive definite function on $G_{\mathsf{d}}$.
\end{itemize}
\end{proposition}

\noindent{\bf Proof}: The equivalence $(i) \Longleftrightarrow (ii)$
is \cite[Thm.~10.1]{ARMA}.  The  equivalence $(iii)
\Longleftrightarrow (iv)$ follows immediately from the definition of
positive definiteness (as the underlying topology is not relevant
for the definition).  To complete the proof we show that $(ii)$ and
$(iv)$ are equivalent.

We first prove the $(iv) \Longrightarrow (ii)$: As $f$ is a positive
definite function on $G_{\mathsf{d}}$ and the Haar measure
$\theta_{G_{\mathsf{d}}}$ is a positive definite measure on
$G_{\mathsf{d}}$, it follows from \cite[Cor.~4.3]{ARMA1} that $f
\theta_{G_{\mathsf{d}}}$ is a positive definite measure on
$G_{\mathsf{d}}$. As $\mu = f \theta_{G_{\mathsf{d}}}$, this proves
$(ii)$.

We next prove $(ii) \Longrightarrow (iv)$: Since $\mu$ is a positive
definite measure on $G_{\mathsf{d}}$, it follows that $\mu*g
*\widetilde{g}$ is a positive definite function for all $g \in
C_{\mathsf{c}}(G_{\mathsf{d}})$.  Let us observe that $g \in
C_{\mathsf{c}}(G_{\mathsf{d}})$ if and only if $g$ has a finite
support. Therefore
\[
    g(x)\,=\, \begin{cases}
              1, & \text{if $x=0$,} \\
              0, & \text{otherwise.}
               \end{cases} \, \in C_{\mathsf{c}}(G_{\mathsf{d}}) \ts .
\]
Our claim now follows from the observation that with this choice of
$g$ we have $f=\mu*g *\widetilde{g}$. \qed

\bigskip

As a consequence of Proposition~\ref{P1} we obtain the following
version of Krein's inequality for positive definite discrete
measures.

\begin{corollary}({\bf Krein's inequality for measures}) Let\/ $\mu$ be a positive definite measure on\/ $G$. Then all\/ $x, t \in G$, we have
\[
   \bigl| \mu(\{ x+t\}) \, - \, \mu(\{ x \}) \bigr|^2
   \,\leq\, 2\ts \mu(\{ 0 \})\ts
   \bigl[ \mu(\{ 0 \}) \, - \, \Re \bigl(\mu(\{ t\})\bigr)\bigr] \ts .
\]
\end{corollary}
\begin{proof} By Lemma \ref{L-help} the measure  $\mu_{\mathsf{pp}}$
is positive definite. By the previous proposition  its support
function $f$ is positive definite. Now, the corollary follows
directly from Proposition \ref{Krein-inequality-for-functions}
applied to $f$.
\end{proof}

We are now going to derive three main consequences of the previous
considerations.

\medskip

Our next result shows  that if $\mu$ is a positive and positive
definite measure on $G$, the set of points of measure close to $\mu
( \{ 0\})\neq 0$ has some sparseness property. This will be a main
ingredient in our study of Meyer sets later on.

\begin{lemma}[Sparseness Lemma]\label{P2}
Let $G$ be an LCAG. Let $\mu$ be a positive and positive definite
translation bounded measure and chose  $ a > (\sqrt{3}-1) \mu(\{
0\})$ and let
\[
I \,:= \, \bigl\{ x \in G \, | \, \mu(\{ x \}) \, \geq \, a \bigr\}
\,.
\]
Then $I-I$ is weakly uniformly discrete.
\end{lemma}
\begin{remark}  Let us note that $\mu(\{0\})$ is  greater than $0$
whenever $\mu$ is not the zero measure (as can easily be inferred
from Krein's inequality).
\end{remark}

\smallskip

\noindent{\bf Proof}: We prove first that there exists some $b
>0$ depending only on $a$ and $\mu(\{ 0 \})$)  such that for all $x,y\in I$
we have $\mu(\{ x-y \}) >b$. By Krein's Inequality we have
\[
   \bigl| \mu(\{ x\}) \, - \, \mu(\{ x-y \}) \bigr|^2
   \,\leq\, 2\ts \mu(\{ 0 \})\ts
   \bigl[ \mu(\{ 0 \}) \, - \, \Re  \bigl(\mu(\{ y\})\bigr)\bigr] \ts .
\]

Therefore, as $\mu$ is positive (and hence real), we have
\[
\begin{split}
   \mu(\{ x-y \}) \, &\geq \, \mu(\{x \}) \, - \,
   \sqrt{ 2\ts \mu(\{ 0 \})\ts
   \bigl[ \mu(\{ 0 \}) \, - \, \Re  \bigl(\mu(\{ y\})\bigr)\bigr]} \\
   &\geq \, a \,-\,
   \sqrt{ 2\ts \mu(\{ 0 \})\ts
   \bigl[ \mu(\{ 0 \}) \, - \, a \bigr]} \ts .
   \end{split}
\]

Let $b=a-\sqrt{ 2\ts \mu(\{ 0 \})\ts \bigl[ \mu(\{ 0 \}) - a
\bigr]}$. A short computation shows that $b>0$ is equivalent to the
condition on $a$ in the statement of the proposition. Indeed,

\begin{eqnarray*}
b > 0 & \Longleftrightarrow &   a  > \sqrt{ 2\, \mu(\{ 0 \}) \ts
    \bigl[ \mu(\{ 0 \}) \, - \,
    a \bigr]}\\
 &\Longleftrightarrow  &
a^2   >    2\, \mu(\{ 0 \})
    \bigl[ \mu(\{ 0 \}) \, - \,
    a \bigr]\\
&\Longleftrightarrow & a^2 \, + \,  2\, \mu(\{ 0 \}) a \, + \,
\bigl( \mu(\{ 0 \}) \bigr)^2 \ts
    > 3 \bigl( \mu(\{ 0 \}) \bigr)^2 \\
    &\Longleftrightarrow &  \bigl[a \bigr.\, +\, \bigl. \mu(\{ 0 \}) \bigr]^2 \ts
    > 3 \bigl( \mu(\{ 0 \}) \bigr)^2\\
     &\Longleftrightarrow &
       a \, >\,
    \mu(\{ 0 \})\bigl( \sqrt{3} \, - \, 1 \bigr).
    \end{eqnarray*}
Now, let
\[
J:= \bigl\{ x \in G \, | \, \mu(\{ x \}) \, \geq \, b \bigr\} \,.
\]

We proved above that $I-I \subset J$. Thus, to complete the proof it
suffices to  show that $J$ is weakly uniformly discrete.

Let $K \subset G$ be any compact set. Since $\mu$ is translation
bounded, we have by definition
\[
C \,:= \, \sup_{t \in G} \mu(t+K)  \, < \, \infty \,.
\]
We show that for all $t \in G$ we have
\[
\sharp \bigl( (t+K) \cap J \bigr) \, \leq \, \frac{ C  } {b } \,,
\]
(which clearly  proves that $J$ is weakly uniformly discrete).
Indeed, for all $t \in G$ we have
\[
C  \geq \mu(t+K) \, \geq \, \sum_{x \in \bigl( (t+K) \cap J \bigr)}
\mu(\{x \}) \, \geq \, b \sharp \bigl( (t+K) \cap J \bigr) \,.
\]
This shows that $J$ is weakly uniformly discrete and the proof is
finished. \qed

\bigskip

Our next two results show that a relatively dense  set can only give
rise to a positive definite Dirac comb if it is a lattice. This ties
in well with various recent strings of research (see remark below).

\begin{lemma}[Rigidity Lemma] \label{L1}
Let $G$ be an LCAG.  Let $\vL \subset G$ be
 weakly uniformly discrete.  Then $\delta_{\vL}$ is
positive definite if and only if $\vL$ is discrete subgroup of $G$.
\end{lemma}
\noindent{\bf Proof}: The implication '$\Longleftarrow$' is obvious.
It remains to prove the implication '$\Longrightarrow$'.  As
$\delta_{\vL}$ is positive definite, it follows from
Proposition~\ref{P1} that the function
\[
    f(x)\,=\, \begin{cases}
              1, & \text{if $x \in \vL$,} \\
              0, & \text{otherwise.}
               \end{cases}  \ts ,
\]
is positive definite. As $f \not\equiv 0$ (as $\vL$ is relatively
dense), it follows that $f(0) \neq 0$, and hence $f(0)=1$. Let now
$x, y \in \vL$ be arbitrary. Then, by Krein's inequality we have:

\[
   \bigl| f(x) \, - \,f(x-y) \bigr|^2
   \,\leq\, 2\ts f(0) \ts
   \bigl[f(0) \, - \, \Re  \bigl(f(y)\bigr)\bigr] =2 [1-1]=0 \ts .
\]

Therefore $f(x-y)=f(x)$. As $f(x)=1$ it follows that $f(x-y)=1$ and
hence $x-y \in \vL$. This shows that
\[
\vL-\vL \subset \vL \ts,
\]
and thus $\vL$ is a subgroup of $G$. As $\vL$ is weakly uniformly
discrete it follows that it must even be uniformly discrete and it
follows  that $\vL$ is a discrete subgroup  in $G$. \qed

\begin{corollary}
Let $\vL \subset G$ be a Delone set. Then $\delta_{\vL}$ is positive
definite if and only if $\vL$ is a lattice.
\end{corollary}
\begin{proof} By the previous lemma, $\vL$ is a discrete subgroup.
By assumption it is furthermore relatively dense and uniformly
discrete. Thus, it is a lattice.
\end{proof}

\begin{remark} The result can be seen in the context
 of a famous theorem of  Cordoba \cite{CORD} and a well-known
 question of Lagarias \cite{LAG}. The theorem of Cordoba
  says that if $\vL$ is a Delone set, and
$\delta_{\vL}$ is Fourier transformable with discrete Fourier
transform, then $\vL$ is crystallographic (i.e. a finite union of
translates of the same lattice). The question of Lagarias asks
whether every Delone set $\vL$ with strongly almost periodic Dirac
comb $\delta_\vL$  is actually crystallographic. Note that if
$\delta_{\vL}$ is Fourier transformable with discrete Fourier
transform, then $\delta_{\vL}$ is a strong almost periodic measure
\cite{ARMA}. Recently a positive answer to Lagarias question was
given under the additional hypothesis of finite local complexity
independently in \cite{FAV} and \cite{KL}. Moreover, in \cite{KL} it
is shown that the answer is in general negative without the
assumption of finite local complexity. In the context of Meyer sets
corresponding results were already obtained in \cite{NS2}.
\end{remark}

Another simple consequence of Proposition~\ref{P1} is the fact that
given a discrete positive definite measure, its restriction to a
closed subgroup of $G$ is also positive definite. In the remainder
of this section   we investigate when the restriction to a subgroup
preserves the positive definiteness. We start by defining the
restriction of a measure to a subgroup.

\begin{definition} Let $G$ be an LCAG.
Let\/ $\mu$ be a measure on\/ $G$ and let\/ $H$ be a closed subgroup
of\/ $G$. We define the restriction of\/ $\mu$ to\/ $H$ by
\[
\mu|_H(B) \, := \, \mu( B \cap H) \ts.
\]
for $B$ a Borel set in $G$.  Then\/ $\mu|_H$ is a measure on\/ $G$
with $\supp(\mu|_H) \subset H$, and can therefore be seen as a
measure on $H$.
\end{definition}

Note that since $H$ is closed in $G$, the characteristic function
$1_H$ is measurable and locally integrable. It is easy to see that
$\mu|_H =1_H \mu$.

\begin{lemma}[Restriction lemma - first version]\label{C1}
Let \/ $\mu$ be a discrete positive definite measure on\/ $G$, and
let\/ $H$ be a closed subgroup of\/ $G$. Then $\mu|_H$ is a positive
definite measure on\/ $H$.
\end{lemma}
\noindent{\bf Proof}: We denote by $G_{\mathsf{d}}$ and
$H_{\mathsf{d}}$, respectively, the groups $G$ and $H$ when equipped
with discrete topology. Let $f: G \to \CC$ be the support function
of $\mu$ given by $f(x) =\mu (\{ x \})$. Then, by
Proposition~\ref{P1}, $f$ is a positive definite function on
$G_{\mathsf{d}}$. This directly gives that the restriction $g : H
\to \CC, g(x)=f(x)$ is a positive definite function on
$H_{\mathsf{d}}$, and hence again by Proposition~\ref{P1} the
measure $\sum_{x \in H_{\mathsf{d}}} g(x) \delta_x$ is positive
definite measure on $H$. But this is exactly the desired statement.
\qed

Combining Proposition~\ref{P1} with Corollary~\ref{C1} we get the
following generalization of \cite[Lemma 8.4]{TAO}:

\begin{corollary}
Let\/ $L$ be a lattice in\/ $G$, let\/ $\eta : L \to \CC$ be a
function and let\/ $\mu = \eta \delta_{L}$. Then\/ $\mu$ is a
positive definite measure on\/ $G$ if and only if\/ $\eta$ is a
positive definite function on\/ $L$.
\end{corollary}

If $H$ is an open subgroup of $G$, it is automatically closed. In
this case, it follows immediately from $C_{\mathsf{c}}(H) \subset
C_{\mathsf{c}}(G)$ that the restriction of any positive definite
measure on $G$ to $H$ is a positive definite measure on $H$.

\begin{lemma}[Restriction lemma - second version]Let $G$ be an LCAG.
Let\/ $\mu$ be a positive definite measure on\/ $G$, and let\/ $H$
be a open subgroup of\/ $G$. Then $\mu|_H$ is a positive definite
measure on\/ $H$.
\end{lemma}
\noindent{\bf Proof}: As $H$ is open in $G$, we have
$C_{\mathsf{c}}(H) \subset C_{\mathsf{c}}(G)$. Therefore, for all $f
\in C_{\mathsf{c}}(H)$ we have
\[
\mu|_H(f*\widetilde{f}) \, = \, \mu(f*\widetilde{f}) \, \geq \, 0
\ts ,
\]
with the first equality follows from the fact that the support of $
f*\widetilde{f}$ is contained in $H$ and the second equality follows
as  $f$ belongs to $ \in C_{\mathsf{c}}(G)$ as well. \qed

\section{On relatively dense sets of $a$-visible Bragg
peaks}\label{section-main} In this section we restrict our attention
to  positive and positive definite measures in $\RR^d$. We will
combine our previous considerations with certain ingredients from
mathematical diffraction theory  to obtain the Meyer property for
certain subsets of the set of Bragg peaks and to provide a new
characterization of the Meyer property.

\bigskip

We will be interested in Meyer sets. There are various
characterizations of Meyer sets in Euclidean space (see e.g.
\cite{LAG1,MOO,MEY}). Here, we will use that a subset $\varGamma$ of
$\RR^d$ is \textit{Meyer} if and only if $\varGamma$  is relatively
dense and $\varGamma - \varGamma$ is weakly uniformly discrete. A
more common definition requires that $\varGamma$ is relatively dense
and $\varGamma - \varGamma$ is uniformly discrete. However, based on
\cite{LAG1}  these two definitions are shown to be equivalent in the
appendix of \cite{BLM}.

Next we will  review  the theory of mathematical diffraction. For
 overviews of this theory  we refer the reader to \cite{BG,Lenz2,Lenz3}.
During the entire section  $\{ A_n \}_n$ will be a \textit{van Hove
sequence} in $\RR^d$ i.e. the $A_n$ will be relatively compact
subsets of $\RR^d$  with
$$|\partial^R A_n|/ |A_n| \to 0, n\to
\infty$$ for all $R>0$. Here,  $|\cdot|$ denotes Lebesgue measure
and, for  $B\subset \RR^d$,  the set  $\partial^R B$ consists of all
$x\in \RR^d$ whose distance from both $B $ and $\RR^d \setminus B$
does not exceed $R$. Obviously, any sequence of balls  (cubes) with
radius (sidelength) tending to $\infty$ is a van Hove sequence. For
a translation bounded measure  $\omega$ on $\RR^d$ we define
$$\gamma_n :=\frac{ \omega|_{A_n} * \widetilde{\omega|_{A_n}}
}{|A_n|}. $$ Here, for a measure $\nu$ we denote by $\nu|_{A}$ the
restriction of $\nu$ to $A$ and by $\widetilde{\nu}$ the measure
with $\widetilde{\nu} (f) = \overline{ \nu(\widetilde{f}) }$.

\begin{definition} Let $\omega$ be a translation bounded measure.
Any cluster point $\gamma$ of the sequence $(\gamma_n)_n$ in the
vague topology  is called \textit{an  autocorrelation} of $\omega$.
\end{definition}

\begin{remark} Let  $\omega$ be a translation bounded measure in
$\RR^d$. Let  $U$ be  an open relatively compact subset of $\RR^d$.
Then, $C:= \sup |\omega| (x+ U)  <\infty$. As shown in \cite{BL} the
space $\mathcal{M}_{U,C}$ of translation bounded measures $\mu$ on
$\RR^d$ with $\sup |\mu| ( x + U) \leq C$ is compact in the vague
topology and all  $\gamma_n$ belong to this space. It follows that
the sequence $\gamma_n$ always has cluster points.
\end{remark}

As is well-known (and not hard to see) any autocorrelation $\gamma$
of a translation bounded $\omega$  is positive definite. For this
reason its Fourier transform $\widehat{\gamma}$ exists and is a
positive measure on the dual group $\widehat{\RR^d}$ of $\RR^d$, see
\cite{ARMA,BL} for further discussion. We call this Fourier
transform \textit{a diffraction measure} for $\omega$. We define the
autocorrelation of a Delone set $\Lambda \subset \RR^d$ to be the
autocorrelation of its Dirac comb $\delta_\Lambda = \sum_{x \in
\Lambda} \delta_x \,.$

\medskip

Let us recall next  the definition of $a$-visible Bragg peaks, see
\cite{NS2} as well.

\begin{definition}
Let $\mu$ be a translation bounded measure on $\RR^d$, and let
$\gamma$ be any autocorrelation of $\mu$. For each $a >0$ we call
\[
I(a) \, := \,  \{ \chi \in \widehat{\RR^d} \, | \,
\widehat{\gamma}(\{ \chi \}) \, \geq \, a \} \,
\]
the set of \textit{$a$-visible Bragg peaks of $\mu$.}
\end{definition}

After this review of diffraction theory we now note the following
consequence of the Sparseness Lemma ~\ref{P2}.

\begin{lemma}\label{C2}
Let $\mu$ be a positive and positive definite translation bounded
measure on $\RR^d$. If the set
\[
I \, := \,  \bigl\{ x \in \RR^d \, | \, \mu(\{ x \}) \, \geq \, a \bigr\} \,.
\]
is relatively dense for some $a > (\sqrt{3}-1) \mu(\{ 0\})$, then
$I$ is a Meyer set.
\end{lemma}

\noindent{\bf Proof}: By Lemma~\ref{P2} the set $I-I$ is weakly
uniformly discrete. As $I$ is also relatively dense, the statement
follows. \qed

\begin{remark}
A natural question is if the lower bound $(\sqrt{3}-1) \mu(\{ 0\})$
can be improved in Lemma~\ref{C2}. We provide an example which shows
that it cannot be decreased under $\frac{1}{2} \mu(\{ 0\})$: Let
$\mu=\delta_{\ZZ}+\delta_{\pi \ZZ}$. Then for all $a >1 =\frac{1}{2}
\mu(\{ 0\})$ the set
\[
\bigl\{ x \in \RR \, |  \, \mu(\{ x \}) \, \geq \, a \bigr\} \, = \, \{ 0\} \,.
\]
is not relatively dense.  However, in the case   $a=1$ the set
\[
I\,:=\, \bigl\{ x \in \RR \, | \, \mu(\{ x \}) \, \geq \, a \bigr\} \, = \, \ZZ \, \cup \, \pi \ZZ \,,
\]
is relatively dense and
\[
I \,- \,I \, = \, \ZZ \oplus \pi \ZZ \,,
\]
which is dense in $\RR$.
\end{remark}

We are now proceeding to prove our main result.

\begin{theorem}\label{T1}
(a) Let $\mu$ be a positive translation bounded measure on $\RR^d$
and let $\gamma$ be an autocorrelation of $\mu$. If the set $I(a)$
of $a$ visible Bragg peaks of $\mu$ is relatively dense for some $a>
(\sqrt{3}-1) \widehat{\gamma}(\{ 0 \})$,  then $I(a)$ is a Meyer
set.

(b) Let $\vL$ be a Delone set in $\RR^d$ and let $\gamma$ be an
autocorrelation of $\vL$. If the set $I(a)$ of $a$ visible Bragg
peaks of $\mu$ is relatively dense for some $a> (\sqrt{3}-1)
\widehat{\gamma}(\{ 0 \})$,  then for all $(\sqrt{3}-1)
\widehat{\gamma}(\{ 0 \})< b \leq a$ the set $I(b)$ is a Meyer set.
\end{theorem}
\noindent{\bf Proof}: (a) If $\mu$ is a positive translation bounded
measure on $\RR^d$, any autocorrelation $\gamma$ is positive and
positive definite. Therefore, so is $\widehat{\gamma}$ \cite{ARMA1},
\cite{BF}. Therefore, applying the result of Lemma~\ref{C2} to
$\widehat{\gamma}$ we obtain the first statement.

(b)  This  follows from (a)  as $I(a) \subset I(b)$ if $(\sqrt{3}-1)
\widehat{\gamma}(\{ 0 \})< b \leq a$. \qed.

\begin{remark} Let us put the previous result in perspective.
\begin{itemize}
\item In Theorem~\ref{T1}, if some $I(a)$ is relatively dense, then $I(b)$
is relatively dense for all $b <a$. The same is not necessarily true
for $b >a$ as can be easily seen by considering a variant of the
example given in the remark following Lemma  \ref{C2} above. In
fact, the mentioned example has the desired property (but is not a
Delone set). To obtain a similar feature with  a Delone set, we can
consider
\[
\Lambda= \ZZ \times \ZZ \cup \left( (\frac{1}{2},0) +  \ZZ \times
(\pi \ZZ) \right).
\]

% Consider the following example: Let $\vL = \ZZ \backslash n\ZZ$.
%Then a simple computation shows that the autocorrelation of $\vL$ is
%\[
%\gamma \, = \, \bigl( 1-\frac{2}{n}\bigr) \delta_{\ZZ}+\frac{1}{n} \delta_{n \ZZ} \,.
%\]
%Therefore, the diffraction of $\vL$ is
%\[
%\widehat{\gamma}\, = \, \bigl( 1-\frac{2}{n}\bigr) \delta_{\ZZ}+\frac{1}{n^2} \delta_{\frac{1}{n}\ZZ} \,.
%\]
%Then
%\[
%I(1-\frac{2}{n})\,= \, \ZZ \,,
%\]
%is relatively dense, but $I(b) =\{ 0\}$ for all $1-\frac{2}{n} < b < \widehat{\gamma}(\{ 0\})=1-\frac{2}{n}+\frac{1}{n^2}$.

\item
This result can be compared with a corresponding result when the
underlying set is Meyer itself. Then, for each $0 < a<
\widehat{\gamma}(\{ 0\})$ the set
\[
I(a)\,=\, \{ \chi \in \widehat{\RR^d} \, | \, \widehat{\gamma}(\{
\chi \}) \, \geq \,  a \} \,,
\]
of $a$-visible Bragg peaks is  Meyer \cite{NS2,NS5}.
\end{itemize}

\end{remark}

We finish this section by using the preceding results to provide a
new characterization of Meyer sets in terms of positive definite
measures.

\begin{theorem}
Let\/ $\vL \subset \RR^d$ be relatively dense. Then, the following
assertions  are equivalent.
    \begin{itemize}
  \item[$(i)$]$\vL$ is a Meyer set.
  \item[$(ii)$] For each $0 < \varepsilon < 1$ there exists a positive and positive definite measure $\mu$ such that, for all $x \in \vL$ we have
  \[
  \mu(\{ x \}) \, > \, \varepsilon \mu(\{ 0 \}) \,.
  \]
  \item[$(iii)$] There exists a positive and positive definite measure $\mu$ and some $ \sqrt{3}-1 < \epsilon <1$ such that, for all $x \in \vL$ we have
  \[
  \mu(\{ x \}) \, > \, \varepsilon \mu(\{ 0\}) \,.
  \]
  \item[$(iv)$] For each $0 < \varepsilon < 1$ there exists a Meyer set $\Gamma \subset \widehat{\RR^d}$, with autocorrelation $\gamma$, such that
  \[
  \vL \, \subset \, I( \varepsilon \widehat{\gamma}(\{0 \}) \,,
  \]
where $I\bigl( \varepsilon \widehat{\gamma}(\{0 \})\bigr)$ is the set of $\varepsilon \widehat{\gamma}(\{0 \})$-visible peaks of $\Gamma$.
  \item[$(v)$] There exists some $0 < \varepsilon < 1$ and a Meyer set $\Gamma \subset \widehat{\RR^d}$, with autocorrelation $\gamma$, such that
  \[
  \vL \, \subset \, I( \varepsilon \widehat{\gamma}(\{0 \}) \,.
  \]
  \item[$(vi)$] For each $0 < \varepsilon < 1$ there exists a Delone set $\Gamma \subset \widehat{\RR^d}$, with autocorrelation $\gamma$, such that
  \[
  \vL \, \subset \, I( \varepsilon \widehat{\gamma}(\{0 \}) \,.
  \]
  \item[$(vii)$] There exists some $\sqrt{3}-1 < \varepsilon < 1$ and a Delone set $\Gamma \subset \widehat{\RR^d}$, with autocorrelation $\gamma$, such that
  \[
  \vL \, \subset \, I( \varepsilon \widehat{\gamma}(\{0 \}) \,.
  \]
\end{itemize}
\end{theorem}
\noindent{\bf Proof}: For any subset $\Sigma$ of $\RR^d$ and any
$\varepsilon >0$ we define
$$\Sigma^{\varepsilon}:=\{ \chi \in \widehat{\RR^d} : |\chi (x) - 1|
< \varepsilon\;\: \mbox{for all $x\in \Sigma$}\}.$$ Similarly, we
define for any subset  $\Xi $ of $\widehat{\RR^d}$ and any
$\varepsilon >0$
$$\Xi^{\varepsilon}:=\{ x\in {\RR^d} : |\chi (x) - 1|
< \varepsilon\;\: \mbox{for all $\chi\in \Sigma$}\}.$$ We will use
below that Meyer sets can be characterized via these sets.

\medskip

 The implications $(ii) \Longrightarrow
(iii)\, , \, (iv) \Longrightarrow (v)$ and $(iv) \Longrightarrow
(vi) \Longrightarrow (vii)$ are obvious, while $(iii)
\Longrightarrow (i)$ follows from Lemma~\ref{C2}. The implication
$(vii) \Longrightarrow (i)$ follows from Theorem~\ref{T1}.  $(iv)
\Longrightarrow (ii)$ follows from the fact that $\mu
=\widehat{\gamma}$ is positive and positive definite \cite{ARMA,BF}.
 $(v) \Longrightarrow (i)$ follows from \cite[Thm~5.3 (iii)]{NS2}.
To complete the proof we prove $(i) \Longrightarrow (iv)$.

Let $\varepsilon \in (0,1)$ and let $\varepsilon'=1-\varepsilon$ and
$\vL'=\left(\vL^{\frac{\varepsilon'}{2}}\right)^{\frac{\varepsilon'}{2}}$.
Then $\vL \subset \vL'$ and
$\left((\vL')^{\frac{\varepsilon'}{2}}\right)^{
\frac{\varepsilon'}{2}}=\vL'$ \cite{MOO}.  Let $\vG =
(\vL')^{\frac{\varepsilon'}{2}}$. We prove that $\Gamma$ has the
desired property.  Let $\gamma$ be an autocorrelation of $\vG$. Then
$\supp(\gamma) \subset \vG-\vG=: \Delta$.  Therefore, by
\cite[Thm.~3.1]{NS2}, for all $y \in \Delta^{\varepsilon'}$ and all
$x \in \RR^d$ we have
\[
\bigl| \widehat{\gamma}(\{ x+y\}) \, - \, \widehat{\gamma}(\{ x \}) \bigr| \, \leq \, \varepsilon' \widehat{\gamma}(\{0 \}) \ts .
\]
Now, by \cite[proof of Cor.~6.8]{MOO} we have
\[
\vG^{\frac{\epsilon'}{2}} \, \subset \, (\vG-\vG)^{\varepsilon'} \, =\, \Delta^{\varepsilon'} \,.
\]
This implies
\[
\vL \, \subset \,  \vL' \, \subset \,  \vG^{\frac{\epsilon'}{2}} \,
\subset \, \Delta^{\varepsilon'} \,.
\]
 Therefore, for all $y \in \vL \subset
\Delta^{\varepsilon'}$, we have
\[
\bigl| \widehat{\gamma}(\{ y\}) \, - \, \widehat{\gamma}(\{ 0 \}) \bigr| \, \leq \, \varepsilon' \widehat{\gamma}(\{0 \}) \,.
\]
This gives
\[
 \widehat{\gamma}(\{ y\}) \, \geq \, \bigl( 1-\varepsilon' \bigr) \widehat{\gamma}(\{0 \})=\varepsilon \widehat{\gamma}(\{0
 \}),
\]
which  finishes the proof. \qed

\end{document}